\documentclass[aps,pra,amsmath,amssymb,superscriptaddress,notitlepage,10pt,reprint,longbibliography]{revtex4-1}
\usepackage{amsfonts,amsthm,enumerate,times}
\usepackage{mathtools}
\usepackage[usenames,dvipsnames]{color}
\usepackage{hyperref}
\usepackage{graphicx}
\usepackage{float}
\usepackage{etoolbox}
\usepackage{fourier}
\usepackage{bbm}

 \newtheorem{theorem}{Theorem}

 \newtheorem{lemma}[theorem]{Lemma}
 
\definecolor{henrik}{rgb}{.8,.3,0}

\newcommand{\mc}[1]{\mathcal{#1}}

\newcommand{\mb}[1]{\mathbb{#1}}

\newcommand{\e}{\mathrm{e}}

\newcommand{\rmd}{\mathrm{d}}
\newcommand{\tr}{\mathrm{Tr}} 
\newcommand{\Tr}{\mathrm{Tr}} 

\newcommand{\id}{\mathbbm{1}}

\newcommand{\mcA}{\mc{A}}

\newcommand{\ZZ}{\mb{Z}}

\newcommand{\norm}[1]{\left\Vert #1 \right\Vert}

\newcommand{\ket}[1]{\left.\left|{#1}\right.\right\rangle}

\newcommand{\bra}[1]{\left.\left\langle{#1}\right.\right|}

\newcommand{\ketbra}[2]{\ket{#1} \!\! \bra{#2}}

  \newcommand{\proj}[1]{\ketbra{#1}{#1}}

\newcommand{\eps}{\varepsilon}
\newcommand{\diam}{\mathrm{diam}}

\begin{document}
\title{Lieb-Robinson bounds imply locality of interactions}
\author{Henrik Wilming}
\affiliation{Institute  for  Theoretical  Physics,  ETH  Zurich,  8093  Zurich,  Switzerland}
\author{Albert H. Werner}
\affiliation{QMATH, Department of Mathematical Sciences, University of Copenhagen,
Universitetsparken 5, 2100 Copenhagen, Denmark}
\affiliation{NBIA, Niels Bohr Institute, University of Copenhagen, Blegdamsvej 17, 2100 Copenhagen, Denmark}

\begin{abstract}
Discrete lattice models are a cornerstone of quantum many-body physics. They arise as effective descriptions of condensed matter systems and lattice-regularized quantum field theories. Lieb-Robinson bounds imply that if the degrees of freedom at each lattice site only interact locally with each other, correlations can only propagate with a finite group velocity through the lattice, similarly to a light cone in relativistic systems.
Here we show that Lieb-Robinson bounds are equivalent to the locality of the interactions: a system with $k$-body interactions fulfills Lieb-Robinson bounds in exponential form if and only if the underlying interactions decay exponentially in space. In particular, 
our result already follows from the behaviour of two-point correlation functions for single-site observables and generalizes to different decay behaviours as well as fermionic lattice models. As a side-result, we thus find that Lieb-Robinson bounds for single-site observables imply Lieb-Robinson bounds for bounded observables with arbitrary support.
\end{abstract}
\maketitle
A crucial feature of the dynamics of a quantum many-body system with local Hamiltonian is the finite group-velocity $v$ for the spread of correlations first shown by Lieb and Robinson \cite{Lieb1972}.
If $A$ and $B$ are local observables at locations $x$ and $y$ of a lattice, this finite group-velocity implies that two-point correlation functions of the form
\begin{align}
	C_{x,y}(t_1,t_0)= i\!\bra{\Psi}\![A(t_1),B(t_0)]\!\ket{\Psi} \label{eq:two-point}
\end{align}
are exponentially small if $(x, t_1)$ lies outside the effective "light-cone" emanating from $(y,t_0)$ with the \emph{Lieb-Robinson velocity} $v$, see Fig.~\ref{fig:LRbounds}.
Such behaviour is reminiscent of the strict light cone in relativistic field theories and holds for any lattice Hamiltonian with bounded $k$-body interactions that decay exponentially in space.
It is the basis for a host of important results in condensed matter and quantum many-body physics.
To name just a few, these range from the exponential clustering theorem for ground states of gapped Hamiltonians \cite{Hastings2004,Hastings2004a,Hastings2006,Nachtergaele2006}, to the quantization of Hall-conductance \cite{Hastings2014,Bachmann2018}, Lieb-Schultz-Mattis theorems \cite{Hastings2004,Nachtergaele2007a}, Goldstone's theorem \cite{Wreszinski1976,Wreszinski1987}, and the Area Law in one spatial dimension \cite{Hastings2007}.
Besides these theoretical results, the finite group velocity is also crucial to be able to simulate the dynamics of complex quantum systems both on a classical as well as on quantum computers, because it ensures that the dynamics can be well approximate by a local quantum circuit. 
Many-body systems can nowadays be simulated in experiments, for example using ultra-cold atoms in optical lattices, and the Lieb-Robinson cone can be detected experimentally \cite{Cheneau2012,Jurcevic2014,Richerme2014}.

In relativistic field theories, which show a strict light-cone, the interactions are usually required to be strictly local.
A natural question one may therefore ask is whether a sufficiently quick decay of the two-point functions \eqref{eq:two-point} conversely also implies that the dynamics can be generated by a Hamiltonian with local interactions.
If true, such a result would, on the one hand, make the connection with relativistic field theories stronger and hence provide valuable insight into the basic properties of quantum many-body systems.  
On the other hand, it would also have practical consequences for system identification and process tomography, which aims at reconstructing the Hamiltonian of a system from measurements of few  observables.  
This constitutes an important task both from an experimental point of view, but also in order to certify the functioning of quantum information processing devices.
In fact, it was shown recently that the Hamiltonian may be reconstructed efficiently using local measurements provided one has the promise that the interactions are indeed local \cite{Bairey2019,Li2020,Bairey2020,Anshu2020}.

In the following, we show such a connection between correlations functions and the locality of the Hamiltonian:
 If the two-point functions $C_{x,y}(t_1,t_0)$ decay at least as quickly as promised by state of the art Lieb-Robinson bounds, then the interactions decay exponentially in space. More generally, our result shows a tight connection between the decay-behaviour of $C_{x,y}(t_1,t_0)$ for short times and the locality of the interactions.

Importantly, to reach our conclusion, we only require knowledge about the decay of two-point functions of observables supported on single lattice-sites.
In practice, these are the relevant quantities in condensed-matter systems as scattering experiments effectively measure two-point correlation functions.
However, exponentially decaying interactions in turn imply Lieb-Robinson bounds for all observables. Hence, as a by-product, we thus prove that Lieb-Robinson bounds for on-site observables already imply Lieb-Robinson bounds for bounded observables with arbitrary support without the need for a support-size dependent correction (as a direct expansion in terms of a product-operator-basis would require; see Appendix~\ref{app:LRbasis}). 

{\bfseries Set-up and Lieb-Robinson bounds.} Before we state our result, we set up some notation and review the statement of Lieb-Robinson bounds. We consider a lattice $\ZZ^D$ and associate to each point $x\in\ZZ^D$ a $d$-dimensional Hilbert-space $\mc H_x$. We choose the hypercubic lattice $\ZZ^D$ purely for simplicity and definiteness -- what follows is true for any regular lattice and with slight modifications also on more general graphs.
Denote by $\mc A_X$ the set (algebra) of local observables $\mc A_X \simeq \mc B(\mc H_X)$ supported in a region $X$. Here, $\mc H_X = \otimes_{x\in X}\mc H_x$ denotes the Hilbert-space associated to $X$.
If $Y\subseteq X$, then $\mc A_Y$ is naturally embedded in $\mc A_X$ by tensoring with identities on the complement of $Y$ in $X$.
We will mostly omit such identities throughout the paper.
Conversely, we can consider the reduction map from $\mc A_\Lambda$ to $\mc A_X$ with $X\subseteq \Lambda$,  which for any $A\in\mc A_X$ and $B\in \mc A_{X^c}$ with $X^c = \Lambda\setminus X$ is given by
\begin{align}
	\Gamma^\Lambda_X[A B] = \Tr\left[B \mathbbm I_{X^c}\right]\, A,
\end{align}
where, as usually, we identify $B\in \mc A_{X^c}$ with $B\otimes \id_X\in \mc A_\Lambda$ and $\mathbbm I_{X^c}:= \id_{X^c}/d_{X^c}$ corresponds to the maximally mixed state on $X^c$.
In the case $X=\emptyset$, we simply obtain $\Gamma^\Lambda_\emptyset[B] = \id \Tr[B \mathbbm I_{X^c}]$.

It is customary to specify the dynamics of many-body systems through the notion of an \emph{(interaction) potential} $\Phi$.
This is a function associating to each finite subset $X\subset \ZZ^D$ of lattice sites a bounded operator $\Phi(X)=\Phi(X)^\dagger \in\mc A_X$.
For a large, but finite subset $\Lambda\subset \ZZ^D$, we then write
\begin{align}
	H_\Lambda := \sum_{X\subseteq \Lambda}\Phi(X)
\end{align}
for the Hamiltonian on $\Lambda$ induced by the potential. By $\tau_t^\Lambda(\cdot):= \e^{\mathrm i H_\Lambda t}\,\cdot\,\e^{-\mathrm i H_\Lambda t}$ we denote the unitary propagator generated by $H_\Lambda$, which gives rise to the Heisenberg-dynamics of observables in $\mc A_\Lambda$.

The potentials appearing in actual models usually have the property that they only couple at most $k$ spins directly, where $k$ is some finite number independent of the system-size. The prototypical examples are spin-spin-interactions of the Ising-or Heisenberg-type or Coulomb interactions for charged particles instead of spins. In both cases we have $k=2$, but larger values of $k$ appear, for example, as plaquette operators in lattice gauge theories, in the context of topological order, e.g. the toric-code, or as effective interactions in perturbation theory \cite{Kitaev2003,Kempe2004,kitaev2006anyons,Jordan2008}.
We hence call a potential \emph{$k$-body} if
\begin{align}
	\Phi(X) = 0,\quad \text{if}\quad |X|>k,
\end{align}
where $|X|$ denotes the number of lattice sites in $X$.
We emphasize that this condition does not put any restrictions on the range or locality of the interaction.
Even with $k=2$, we, a priori, still allow for arbitrarily long-ranged interactions.
Finally, we call a potential \emph{exponentially decaying} if there exist positive constants $K$ and $a$ such that
\begin{align}
	\norm{\Phi(X)}\leq K \exp(- a\mathrm{diam}(X))
\end{align}
for any region $X\subset \ZZ^D$. Here, $\mathrm{diam}(X)=\sup_{x,y\in X} |x-y|$ denotes the diameter of $X$ with $|x-y|$ the (graph-theoretical) distance on the lattice.

Even though a given potential $\Phi$ gives rise to a fixed time-evolution $\tau^\Lambda_t$ for every $\Lambda$, the converse does not hold true:
Firstly, the Hamiltonian which generates the time-evolution $\tau^\Lambda_t$ is only fixed up to a shift of all the energies. Secondly, a given Hamiltonian can in general be decomposed in many different ways into a potential.
The only physically relevant object in the end is the time-evolution operator $\tau^\Lambda_t$.

Accordingly, the question we are addressing in this paper, namely whether the interaction is exponentially decaying or not, corresponds to asking whether there exists \emph{some} exponentially decaying potential giving rise to the same time-evolution.
Our result shows that given \emph{some} $k$-body potential for which the two-point functions $C_{x,y}(t_1,t_0)$ of the corresponding time-evolution show a certain decay-behaviour in space, then there also exists a $k$-body potential on $\Lambda$ generating exactly the same time-evolution, which has the same decay-behaviour in space.
\begin{figure}[t!]
	\includegraphics[width=7cm]{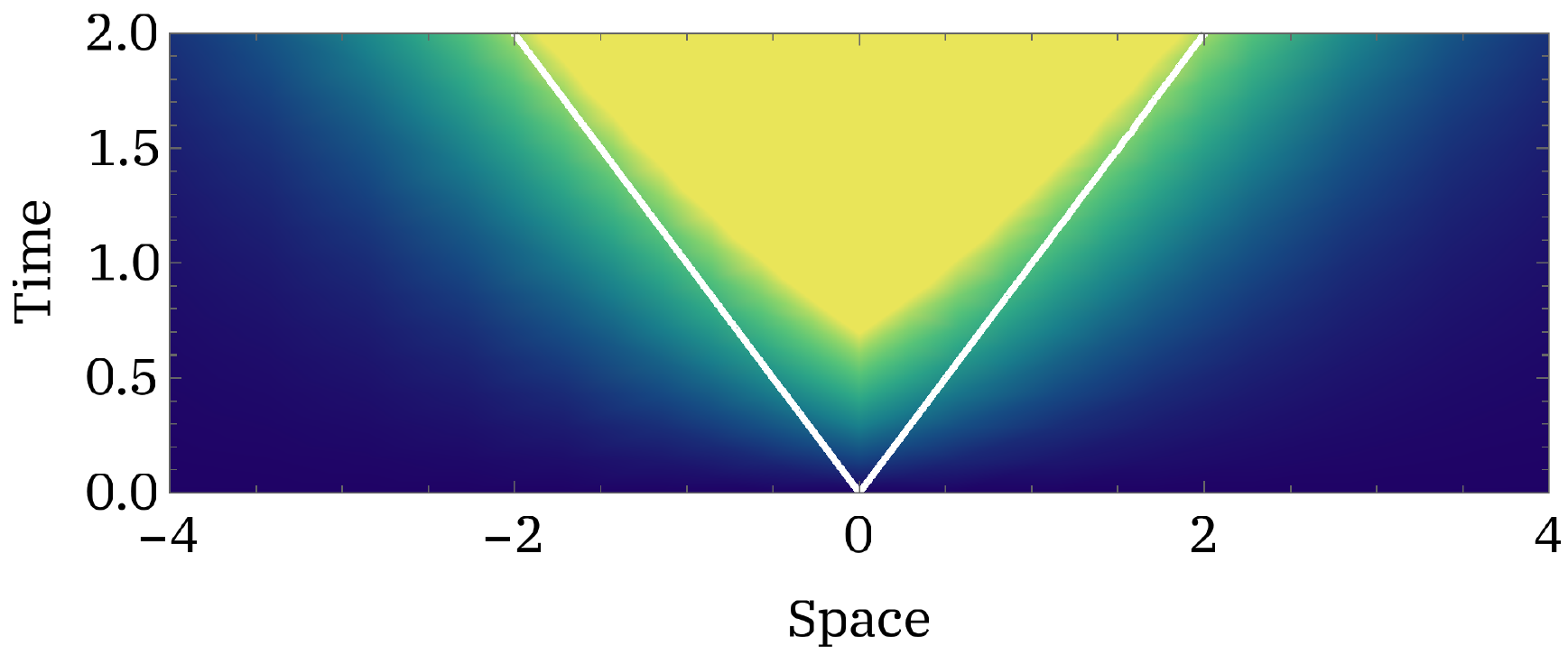}
	\caption{Color: Density-plot of  Lieb-Robinson cone in the form given by Theorem~\ref{thm:LRbounds}. White: Contour of Lieb-Robinson cone as given by the standard exponential form of the Lieb-Robinson bounds. For large times the two agree, but for short times we see a significantly different behaviour.}
	\label{fig:LRbounds}
\end{figure}
Importantly, the following theorem from Ref.~\cite{Nachtergaele2007}, which is an improved version of the original bound from Ref.~\cite{Lieb1972}, states that any exponentially decaying $k$-body potential conversely gives rise to suitably bounded two-point functions. We state the theorem here in a form adjusted to our set-up.
\begin{theorem}[Lieb-Robinson bounds]\label{thm:LRbounds} Given an exponentially decaying $k$-body potential, there exist constants $\mu,v,K > 0$ such that for any $\Lambda,X,Y\subset \ZZ^D$, we have
	\begin{align}
		C^\Lambda_{X,Y}(t)&:=\max_{A,B}\frac{\norm{[\tau^\Lambda_t(A),B]}}{\norm{A}\norm{B}} \leq  2\mathrm{min}\{1,g(t)  f(X,Y)\}, \label{eq:LRbounds}
	\end{align}
	where the maximization is over operators $A\in \mc A_X$ and $B\in \mc A_Y$ supported in $X$ and $Y$, respectively, and $f(X,Y)=\min\{|X|,|Y|\} K \exp(-\mu d(X,Y))$. The function $g(t)$ is given by
	\begin{align}
	g(t) = \begin{cases}
		\exp(\mu v\, |t|) -1\quad &\text{if}\quad d(X,Y)>0\\
		\exp(\mu v\,|t|) \quad &\text{if}\quad d(X,Y)=0.
		\end{cases}
	\end{align}
\end{theorem}
Using the fact that the operator norm is invariant under unitary transformations and
$[A(t_1),B(t_0)]$ = $\tau^\Lambda_{t_0}([\tau^\Lambda_{t_1-t_0}(A),B])$, Lieb-Robinson bounds imply that the two-point correlation functions in \eqref{eq:two-point} satisfy
\begin{align}
	|\!\bra{\Psi}\![A(t_1),B(t_0)]\!\ket{\Psi}\!|\leq \min\big\{2,K (\e^{\mu v |t_1-t_0|}-1) \e^{-\mu |x-y|}\big\}\nonumber
\end{align}
for any state $\ket \Psi$ and observables with $\norm{A}=\norm{B}=1$.
The bound is plotted in Figure~\ref{fig:LRbounds} and compared with the conventional bound $\exp(\mu(v|t_1-t_0| - |x-y|))$.
For most purposes this latter bound is good enough, since it agrees for large times with the one we use here.
However, for short times, the two bounds differ markedly. In particular, the above version ensures that $C^\Lambda_{X,Y}(0) = 0$ for $d(X,Y)>0$ and this will be crucial for deriving our result.
The functions $g(t)$ and $f(X,Y)$ specify how sharp the Lieb-Robinson cone is localized.
If the potential decays slower than exponentially, for example as a power-law, then they have to be replaced by functions that decay slower than exponentially.
Indeed, recently a fair amount of literature is devoted to studying Lieb-Robinson bounds for long-ranged interactions, see, e.g., Refs.~\cite{Hastings2006,Eisert2013,Hauke2013,Gong2014,FossFeig2015,Matsuta2016,Kuwahara2019a,Chen2019,Else2020,Tran2020}. 
Furthermore, let us mention that the commutator form of the Lieb-Robinson bound given in \eqref{eq:LRbounds} can be reformulated in two ways. Firstly, it gives a bound on the detectability of local excitations at distant lattice sites. Secondly, it provides a certificate on how well the evolution of a local observable can be approximated by an evolution under a truncated Hamiltonian. 
In particular, the first formulation can be interpreted as a bound on the information that can be communicated from one part of the lattice to another part using the given time-evolution: it implies that if one tries to send information with non-vanishing capacity using a local encoding of the information at one point of the lattice, then there is a minimal time one has to wait before the information can be decoded at a different part of a lattice.

{\bfseries Main result.} Our main result is a direct converse to the Lieb-Robinson bounds stated in Theorem~\ref{thm:LRbounds}. Importantly, however, it only depends on the behaviour of two-point functions for observables located at single lattice-sites.
\begin{theorem}[Converse to Lieb-Robinson bounds]\label{thm:converse}
	Consider a $k$-body potential and assume that there exists a function $h(t,x)$, with $h(0,x)=0$ for $|x|>0$ and differentiable in $t$, such that for any finite subset $\Lambda\subset \ZZ^D$ and any two points $x\neq y\in \ZZ^D$ we have
	\begin{align}\label{eq:converse} 
		C^\Lambda_{x,y}(t) :=\max_{A,B} \frac{\norm{[\tau^\Lambda_t(A),B]}}{\norm{A}\norm{B}}  \leq h(t,|x-y|),
	\end{align}
	where the maximization is over single-site operators situated at $x$ and $y$: $A\in \mc A_{\{x\}}, B\in \mc A_{\{y\}}$.
	Then there exists a constant $m_k$ only depending on $k$, and a $k$-body potential $\hat \Phi_\Lambda$ on $\Lambda$, giving rise to the same time-evolution $\tau^\Lambda_t$, with
	\begin{align}
		\norm{\hat\Phi_\Lambda(X)} \leq m_k\, \partial_{\!t}\!h\big(0,\mathrm{diam}(X)\big).
	\end{align}
	for all $X\subseteq\Lambda$ with $|X|\geq 2$.
\end{theorem}
Inserting $h(t,x) = g(t)f(x)$ with $f(x) = K \exp(-\mu x)$ and $g(t)$ from the Lieb-Robinson bounds into the theorem, the result shows directly that Lieb-Robinson bounds for two-point functions imply the exponential decay of the interaction potential $\hat\Phi_\Lambda$:
\begin{align}
	\norm{\hat\Phi_\Lambda(X)} \leq m_k \mu v K \exp\big(-\mu \mathrm{diam}(X)\big).
\end{align}
Moreover, since exponential decay of the interaction implies the Lieb-Robinson bounds in Theorem~\ref{thm:LRbounds}, we find that the decay of two-point functions of single-site observables implies a corresponding decay for arbitrary observables.

More generally, a uniform bound on the decay behaviour of two-point functions directly translates into a decay-behaviour for the potential. For example, if we obtain a power-law decay for the two-point functions, we also establish a power-law behaviour for the potential.

We may also infer strict locality of the Hamiltonian from the function $h$. Suppose, for example, that for $|x-y|>R$, $h(t,|x-y|) = t^{\alpha}f(|x-y|)$ with $\alpha>1$. Then $\partial_t h(0,\mathrm{diam}(X))=0$ if $\mathrm{diam}(X)>R$ and therefore $\hat \Phi_\Lambda(X)$ vanishes.

Note, that we only obtain bounds for potential terms with $|X|\geq 2$. That is, we do not get bounds on the norm of on-site potentials, such as magnetic fields. This is in accordance with the fact that Lieb-Robinson bounds can also be proven for Hamiltonians with unbounded on-site potentials \cite{Nachtergaele2008,Nachtergaele2014}. On the other hand, there are Hamiltonians on bosonic lattices with local but unbounded interactions that allow for signalling at arbitrary speeds \cite{Eisert2009}. This shows that interactions have to be appropriately bounded to obtain Lieb-Robinson bounds for arbitrary observables, in agreement with our result.

Finally, let us consider the possibility of strictly causal behaviour in lattice systems: 
Let us assume that $C^\Lambda_{x,y}(t)$ vanishes identically whenever $v t<|x-y|$ for some $v$. Then our result directly implies that the interaction terms involving distinct lattice sites vanish identically as well. In other words, the only dynamics on a lattice system, which is strictly causal for all times, is one where the spins do not interact at all.

{\bf The proof.}
A central ingredient in our proof is the reduction map $\Gamma^\Lambda_X$, which we use to define $\hat \Phi_\Lambda$ as a certain \emph{canonical form} of a given potential $\Phi$ (see also Refs.~\cite{Araki2004,BratteliRobinson2}):
We say that a potential $\hat\Phi_\Lambda$ on $\Lambda$ is in \emph{canonical form} if
	\begin{align}	
		\Gamma^\Lambda_X[\hat\Phi_\Lambda(Y)] &= 0
	\end{align}
unless $Y\subseteq X$.
Given a Hamiltonian $H_\Lambda$ on $\Lambda$, we can always decompose $H_\Lambda$ into a potential in canonical form. Moreover, if $H_\Lambda$ can be written in terms of a $k$-body potential, then the resulting canonical form will also be $k$-body. These facts are collected in the following Lemma and proven in the Supplementary Material.
\begin{lemma}\label{lemma:canonical}
Consider a $k$-body potential $\Phi$ with associated Hamiltonians $H_\Lambda$. Then for every finite subset $\Lambda\subseteq \ZZ^D$, the potential defined by
\begin{align}
	\hat \Phi_\Lambda(Z) &= \sum_{X\subseteq Z} (-1)^{|Z|-|X|}\Gamma^\Lambda_X[H_\Lambda],
\end{align}
is a $k$-body potential in canonical form on $\Lambda$ giving rise to the Hamiltonian $H_\Lambda$.
\end{lemma}
Given this Lemma, we are now ready to give the essential steps of the proof of Theorem~\ref{thm:converse}. To lighten the notation, we now fix $\Lambda$ and omit it as sub- or superscript from the potential, the reduction map, the propagator and the quantity $C_{x,y}(t)$.
Let us consider two points $x\neq y\in \Lambda$ and operators $A\in \mc A_{\{x\}}$ and $B\in \mc A_{\{y\}}$ with $\norm{A}=\norm{B}=1$.  Since $[A,B]=0$, we then have
\begin{align}
	\lim_{\delta t\rightarrow 0}\frac{\norm{[\tau_{\delta t}(A),B]}}{\delta t}  &= \norm{[[H_\Lambda ,A],B]}.
\end{align}
On the other hand, we can make use of the assumption on $C_{x,y}(t)$ to find
\begin{align}
	\lim_{\delta t\rightarrow 0}\frac{\norm{[\tau_{\delta t}(A),B]}}{\delta t}
	&\leq \lim_{\delta t\rightarrow 0} \frac{C_{x,y}(\delta t)}{\delta t} \\
	&\leq \lim_{\delta t\rightarrow 0} \frac{h(\delta t,|x-y|)}{\delta t}\\
	&= \partial_{\!t}\! h(0,|x-y|),
\end{align}
where we used $h(0,x)=0$ in the last step.
We thus finally get
\begin{align}\label{eq:basic}
	\norm{\left[\left[H_\Lambda,A\right],B\right]}\leq \partial_{\!t}\! h(0,|x-y|).
\end{align}
In the following, let us write $\eps_{x,y} := \partial_{\!t}\! h(0,|x-y|)$ and $x^c := \Lambda\setminus\{x\}$.
Since \eqref{eq:basic} holds for all $B\in \mc A_{\{y\}}$ with $\norm{B}=1$, a Lemma from Ref.~\cite{Bachmann2011}, whose proof we provide in the Supplementary Material (Lemma~\ref{lemma:reductionspins}), implies
\begin{align}\label{eq:firstrestriction}
	\norm{\left[(1-\Gamma_{y^c})[H_\Lambda],A\right]}\leq \eps_{x,y}.
\end{align}
Here, we used $\Gamma_{y^c}[H_\Lambda A]=\Gamma_{y^c}[H_\Lambda]A$ and  $\Gamma_{y^c}[AH_\Lambda]=A\Gamma_{y^c}[H_\Lambda]$, which holds because $A$ is supported within $y^c$. 
Since \eqref{eq:firstrestriction} again holds for all normalized $A\in\mc A_{\{x\}}$, we can use the same argument again to obtain
\begin{align}\label{eq:bound}
	\norm{(1-\Gamma_{x^c})[(1-\Gamma_{y^c})[H_\Lambda]]}\leq \eps_{x,y}.
\end{align}
Now let $\hat\Phi$ be the potential in canonical form representing the Hamiltonian $H_\Lambda$.
We then have $\Gamma_X[H_\Lambda]=H_X$ and $\hat\Phi(\emptyset)\propto \mathbf 1$ (and the identity-operator only appears in this term).
Thus,
using $\Gamma_{x^c}\circ \Gamma_{y^c} = \Gamma_{x^c \cap y^c}$,
\begin{align}
	(1-\Gamma_{x^c})\left[(1-\Gamma_{y^c})[H_\Lambda]\right] &= H_\Lambda - H_{y^c} - H_{x^c} + H_{y^c\cap x^c}\\
	&=\sum_{Z \ni x,y} \hat\Phi(Z).
\end{align}
Therefore, \eqref{eq:bound} says that for any pair of lattice sites $x\neq y$, we have
\begin{align}\label{eq:step2}
	\Big\| \!\! \sum_{Z'\ni x,y}\hat\Phi(Z') \Big\| \leq \eps_{x,y}.
\end{align}
We now fix a set $Z$ and take two points $x,y\in Z$ such that $|x-y|=\mathrm{diam}(Z)$
and make use of the fact that operators can only become smaller in norm under reductions to sub-systems, i.e., for any
$C=C^\dagger \in\mathcal A_\Lambda$ and $X\subseteq \Lambda$ we have
	\begin{align}
		\norm{\Gamma_X[C]}\leq \norm{\Gamma_X[\id]}\norm{C} = \norm{C}, \label{lemma:lowerbound}
	\end{align}
where we used that $\Gamma_X$ is a completely positive, unital map. This property allows us to restrict the bound in \eqref{eq:step2} to potential terms supported only within $Z$:
\begin{align}\label{eq:step3}
	\Big \|\!\!\sum_{\substack{Z'\ni x,y\\Z'\subseteq Z}}\hat\Phi(Z')\Big\| \leq \eps_{x,y}.
\end{align}
Unfortunately, this bound does not yield a bound for each of the terms in the sum by $\eps_{x,y}$ individually. 
This is because we cannot assume that the potential terms are positive operators due to our demand that the potential is in canonical form.  
We therefore now make use of the fact that $|Z|\leq k$ and use the inverse triangle inequality to obtain
\begin{align}\label{eq:step4}
	\Big\|\hat\Phi(Z)\Big\| \leq \eps_{x,y} + \Big\|\!\! \sum_{\substack{Z'\ni x,y\\Z'\subset Z}}\hat\Phi(Z') \Big\|.
\end{align}
Since the sum on the r.h.s is over strict subsets of $Z$, each of the potential terms on the r.h.s. couples at most $|Z|-1$ spins.
Clearly, the above equation gives $\norm{\hat\Phi(Z)}\leq\eps_{x,y}$ for $|Z|=2$, since the only possible set is $Z=\{x,y\}$.
Hence consider $|Z|=3$. Then, the norm on the r.h.s. only contains one set $Z'$, which has cardinality $|Z'|=2$. Thus $\norm{\hat\Phi_\Lambda(Z)}\leq 2\eps_{x,y}$.
This reasoning allows us to recursively bound the potential terms $\Phi(Z)$ with $|Z|=m$ in terms of the those with $|Z|=m-1$.  Since we are assuming the interaction potential to be $k$-body, the recursion stops after a constant number of steps. Hence, there is a finite number $m_k$ only depending on $k$ such that
\begin{align}
	\norm{\hat\Phi_\Lambda(Z)} \leq m_k \eps_{x,y} =  m_k\,\partial_{\!t}\!h\big(0,\diam(Z)\big),
\end{align}
where we inserted the definition of $\eps_{x,y}$ and used that $\diam(Z)=|x-y|$.
This finishes the proof of Theorem~\ref{thm:converse}.

{\bfseries Fermionic lattice systems.} So far, we only considered spin systems. In the Supplementary Material we generalize our result to fermionic lattice systems.
The only difference to the case of spins is that we need to allow the observables $A,B$ in the formulation of the (assumed) bound in Theorem~\ref{thm:converse} to act not only on the respective lattice points, but also on potentially present fermionic auxiliary systems. 
Such auxiliary system would make no difference in the case of spins. 
However, due to the anti-commutation relations of fermions and the parity super-selection rule, they are important to consider in the case of a fermionic lattice system.

{\bfseries Discussion. } In this work we have shown that whenever the two-point functions $C_{x,y}(t_1,t_0)$ show a certain decay-behaviour in space, then there also exists a potential generating exactly the same time-evolution, which has the same decay-behaviour in space.
In particular, this shows that Lieb-Robinson bounds in exponential form are equivalent to the exponential decay of interactions.
Our result implies that whenever a system has \emph{some} long-range interactions, then this can be diagnosed by measuring a two-point function in a suitable initial state.
Furthermore, considering the reformulation of Lieb-Robinson bounds in terms of the detectability of local excitations, we see that there is a direct correspondence between the possibility of sending information through the lattice with finite capacity and the locality of the interactions. Let us conclude by discussing some open problems and avenues for further research.

First, we derived our results in the setting of continuous-time, unitary dynamics.
One can ask whether similar results also hold for non-unitary dynamics or when one has discrete-time unitary dynamics.
Suppose, for example, that one only has promises for the time-evolution in discrete time-steps $\delta\! t$, i.e., for the quantities $C^{\Lambda}_{A,B}(j\,\delta\! t)$ with $j$ an integer.
Then the dynamics is given by $\tau_{\delta\! t}(A) =  U A U^*$ (omitting $\Lambda$ for clarity) and $U$ is a quasi-local unitary implementing the time-evolution for a time-step $\delta\! t$.
Compared to large distances on the lattice, this time-step is arbitrarily small.
Going to very large distances, we could therefore imagine that on large scales the dynamics is then generated by a quasi-local Hamiltonian.
Interestingly, according to Refs.~\cite{Farrelly2019,Zimboras202x}, this is not the case: there exist \emph{strictly local} unitary dynamics in discrete time, which cannot  be generated by quasi-local Hamiltonians, despite the fact that they can always be generated by local, finite depth quantum circuits \cite{Arrighi2011}. 
In the supplemental material, we present a very simple example of this behaviour \footnote{We thank Terry Farrelly and Zolt\'an Zimbor\'as for reminding us of this example.}. 
We thus conclude that it is indeed necessary to have promises for arbitrary short times.
Whether our results generalize to continuous-time, dissipative, Markovian time-evolution is an open problem.
Let us remind the reader that Lieb-Robinson bounds also hold for dissipative systems, whose time-evolution is generated by quasi-local Lindbladians \cite{Hastings2004a,Poulin2010,Nachtergaele2011,Barthel2011,Kliesch2014a}.
Here, however, we would have to show that any long-ranged term in a Lindbladian leads to a violation of the corresponding Lieb-Robinson bound.
It's not clear whether this is true for purely dissipative evolution.
Certainly the techniques used in this work would have to be extended to prove such a result.
If a counter-example indeed exists, it would show that there exists a close connection between unitarity of time-evolution and locality of the dynamics in the sense of information propagation, which would be remarkable.
We leave the investigation of this interesting point for future work.

Secondly, our results only need a promise over correlation-functions of single-site observables for very short times and in this sense are rather friendly to experimental settings.
However, to derive our results, we require a promise on the behaviour of such correlation functions for arbitrary states $\ket{\Psi}$.
In practice, such information is usually not available, since we do not prepare arbitrary states in real experiments. On the other hand, we are usually also not interested in the dynamics for arbitrary states, but only in some sub-set of states, for example from the low-energy sector. It is an interesting open problem to find out, whether having a bound of the form
\begin{align}
	\bra{\Psi}[\tau^\Lambda_t(A),B]\ket{\Psi} \leq  \norm{A}\norm{B} g(t) \exp(-\mu |x-y|)
\end{align}
for all states $\ket \Psi$ in some experimentally relevant subset of states implies that the dynamics, for this subset of states, can be represented by an exponentially decaying interaction.

Thirdly, we have shown our results for spin models and fermionic lattice models. 
It would be interesting to investigate, whether similar results as ours can be sensibly formulated and proven for the case of bosonic systems (restricting to bounded operators as observables).
In this setting it is still an open problem to formulate and prove general Lieb-Robinson bounds for interacting bosonic systems (see, however, Refs.~\cite{Nachtergaele2008,Nachtergaele2010,Nachtergaele2014}). 
Indeed, without additional assumptions, they cannot be proven due to the results in Ref.~\cite{Eisert2009}. 
The assumption of reasonably bounded correlation functions naturally rules out interactions that lead to violations of Lieb-Robinson bounds. 
Methods similar to the ones we used in this work may thus help identifying the general form of bosonic interactions for which Lieb-Robinson bounds can be proven.

{\bfseries Acknowledgements. } We would like to thank Terry Farrelly and Zolt\'an Zimbor\'as for interesting comments and discussions. H.~W. acknowledges support through the National Centre of Competence in Research \emph{Quantum Science and Technology} (QSIT). A.~H.~W. thanks the VILLUM FONDEN for its support with a Villum
Young Investigator Grant (Grant No. 25452) and its support via the QMATH Centre of
Excellence (Grant No. 10059).

\bibliographystyle{apsrev4-1}
\bibliography{literatur}

\appendix
\onecolumngrid

\section{Proof of Lemma~\ref{lemma:canonical}}
\label{sec:app:canonical}
In this section, we provide the proof of Lemma~\ref{lemma:canonical}, i.e., the statement that
the potential defined by
\begin{align}
	\hat \Phi_\Lambda(Z) &= \sum_{X\subseteq Z} (-1)^{|Z|-|X|}\Gamma^\Lambda_X[H_\Lambda]
\end{align}
is a $k$-body potential in canonical form on $\Lambda$ giving rise to the Hamiltonian $H_\Lambda$.
This form of the potential is motivated by the \emph{inclusion-exclusion} principle (see, e.g., Ref.~\cite{Gessel1996}).

We need to show the following facts: i) The potential gives rise to the same Hamiltonian $H_\Lambda$, ii) the potential is in canonical form, and
iii) the potential is $k$-body. Before we come to the details of this, we collect some simple facts about sums over subsets.
The crucial fact that we will need is an elementary form of the inclusion-exclusion principle. It states that, for any finite set $Z$, we have
\begin{align}\label{eq:basic-inc-ex}
\sum_{X\subseteq Z}(-1)^{|X|}
	&=  \delta_{\emptyset}(Z)
	= \begin{cases}
	0 \quad\  \text{if}\  Z \neq \emptyset\\
		1 \quad\ \text{if}\  Z = \emptyset.
	\end{cases}
\end{align}
This equality is very simple to prove: Clearly it's true for $Z=\emptyset$. Therefore consider any $Z$ with $Z\neq \emptyset$ and let $z\in Z$ be an arbitrary point. Then
\begin{align}
	\sum_{X\subseteq Z}(-1)^{|X|} &= \sum_{X\subseteq Z, z\in X}(-1)^{|X|} + \sum_{X\subseteq Z, z\notin X}(-1)^{|X|} =\sum_{X\subseteq Z\setminus\{z\}}(-1)^{|X\cup\{z\}|} + \sum_{X\subseteq Z\setminus\{z\}}(-1)^{|X|} = \sum_{X\subseteq Z\setminus\{z\}}(-1 + 1) (-1)^{|X|}  = 0.
\end{align}
We will also use that for any function of sets $f$, we can write 	
\begin{align}\label{eq:id1}
	\sum_{Z: X\subseteq Z \subseteq \Lambda} f(Z) = \sum_{Z\subseteq \Lambda\setminus X} f(Z\cup X).
\end{align}
In particular, we have
\begin{align}\label{eq:id2}
	\sum_{Z: X\subseteq Z \subseteq \Lambda}(-1)^{|Z|} = \sum_{Z\subseteq \Lambda\setminus X} (-1)^{|Z|+|X|} = (-1)^{|X|} \delta_\emptyset(\Lambda\setminus X).
\end{align}
Finally, we require the identity.
\begin{align}\label{eq:id3}
	\sum_{Z\subseteq \Lambda}\sum_{X\subseteq Z} f(X)g(X,Z) = \sum_{X\subseteq \Lambda}\sum_{Z: X\subseteq Z\subseteq \Lambda}f(X) g(X,Z)= \sum_{X\subseteq \Lambda}f(X) \sum_{Z: X\subseteq Z\subseteq \Lambda} g(X,Z),
\end{align}
which holds for arbitrary functions of sets $f$ and $g$. With these ingredients in place, we now start with the proof of the Lemma.

We first show that the potential sums up to give the Hamiltonian:
\begin{align}
\sum_{Z\subseteq \Lambda} \hat \Phi_{\Lambda}(Z) &= \sum_{Z\subseteq \Lambda}\sum_{X\subseteq Z}(-1)^{|Z|-|X|}\Gamma^\Lambda_X[H_\Lambda] \\
	&= \sum_{X\subseteq \Lambda} \Gamma^\Lambda_X[H_\Lambda]\underbrace{\sum_{Z:X\subseteq Z\subseteq \Lambda} (-1)^{|Z|-|X|}}_{\delta_\emptyset(\Lambda\setminus X)} = \Gamma^\Lambda_\Lambda[H_\Lambda]=H_\Lambda,
\end{align}
where we made use of \eqref{eq:basic-inc-ex} and $\Gamma^\Lambda_\Lambda=1$.
To show that the potential is in canonical form, we consider $Y\subset \Lambda$ and $Z\subset\Lambda$ such that $Z\cap Y^c\neq \emptyset$ (i.e., $Z\setminus (Z\cap Y)\neq\emptyset$). We then find, using $\Gamma^\Lambda_Y \circ \Gamma^\Lambda_X = \Gamma^\Lambda_{X\cap Y}$,
\begin{align}
	\Gamma^\Lambda_Y [\hat\Phi(Z)] &=  \sum_{X\subseteq Z} (-1)^{|Z|-|X|}\Gamma^\Lambda_{X\cap Y}[H_\Lambda].
\end{align}
Let us decompose $Z$ into $Z\cap Y$ and $Z\setminus (Z\cap Y)\neq \emptyset$. For any function of sets $f$, a sum over subsets of $Z$ can be decomposed as
\begin{align}
	\sum_{X\subseteq Z} f(X)= \sum_{X'\subseteq Z\setminus (Z\cap Y)} \sum_{X''\subseteq Z\cap Y} f(X'\cup X'').
\end{align}
We thereby find
\begin{align}
	\Gamma^\Lambda_Y [\hat\Phi_\Lambda(Z)] &=  \underbrace{\sum_{X'\subseteq Z\setminus (Z\cap Y)} (-1)^{|Z|-|X'|}}_{\delta_\emptyset\big(Z\setminus (Z\cap Y)\big)(-1)^{|Z|}=0}\sum_{X''\subseteq Z\cap Y}(-1)^{-|X''|} \Gamma^\Lambda_{X''}[H_\Lambda] = 0.
\end{align}
Finally, let us show that $\hat \Phi_\Lambda$ is $k$-body.
To see this, consider $Z$ with $|Z|>k$. We have to show that $\hat \Phi_\Lambda(Z)=0$.
First, we use the original $k$-body potential $\Phi$ to write
\begin{align}
	\hat\Phi_\Lambda(Z) = \sum_{X\subseteq Z}(-1)^{|Z|-|X|}\Gamma^\Lambda_X\big[\!\!\sum_{\substack{Y\subseteq \Lambda\\ |Y|\leq k}} \Phi(Y)\big] = \sum_{\substack{Y\subseteq \Lambda\\ |Y|\leq k}} \sum_{X\subseteq Z}(-1)^{|Z|-|X|}\Gamma^\Lambda_X[\Phi(Y)].
\end{align}
We now show that for each $Y$ the corresponding term is zero. To do this, we again split the sum over $X$ into a sum over $Z\cap Y$ and $Z\setminus (Z\cap Y)$ to get
\begin{align}
		\sum_{X\subseteq Z}(-1)^{|Z|-|X|}\Gamma^\Lambda_X[\Phi(Y)] &= \sum_{X'\subseteq Z\setminus(Z\cap Y)} (-1)^{|Z|-|X'|}
		\sum_{X''\subseteq Z\cap Y}(-1)^{|X''|} \Gamma_{X'\cup X''}^\Lambda[\Phi(Y)].
	\end{align}
	However, $\Phi(Y)\in \mc A_Y$ and the regions $X'$ do not overlap $Y$ by definition. Accordingly, we have $\Gamma^\Lambda_{X'\cup X''}[\Phi(Y)]= \Gamma^\Lambda_{X''}[\Phi(Y)]$. We thus find	
	\begin{align}
		\sum_{X\subseteq Z}(-1)^{|Z|-|X|}\Gamma^\Lambda_X[\Phi(Y)] &=  \underbrace{\sum_{X'\subseteq Z\setminus (Z\cap Y)} (-1)^{|Z|-|X'|}}_{\delta_\emptyset\big(Z\setminus (Z\cap Y)\big)(-1)^{|Z|}}
		\sum_{X''\subseteq Z\cap Y}(-1)^{|X''|} \Gamma_{X''}^\Lambda[\Phi(Y)]\\
		&= \delta_\emptyset\big(Z\setminus (Z\cap Y)\big)(-1)^{|Z|}
		\sum_{X''\subseteq Z\cap Y}(-1)^{|X''|} \Gamma_{X''}^\Lambda[\Phi(Y)].
	\end{align}
	Now, since $\Phi(Y)$ is zero for $|Y|>k$ the expression vanishes in this case. On the other hand, if $|Y|\leq k$, we can make use of our assumption that $|Z|>k$. This implies that $Z\setminus (Z\cap Y)\neq \emptyset$ and hence the $\delta$-function vanishes. This finishes the proof.

\section{Reduction map \& commutator}

For completeness, we provide in this section a proof of a well-known result about the quality of local restrictions of an observable, given that it almost commutes with all observables on the complement of this restriction \cite{Bachmann2011}.
\begin{lemma}\label{lemma:reductionspins}
	Let $A\in \mathcal{A}_\Lambda$ and assume that
	$\norm{[A,B]}\leq \varepsilon {\norm{B}}$  for all $B\in\mathcal{A}_Y$ for some subset $Y\subset\Lambda$. Then
  \begin{align}
	  \norm{(1-\Gamma_{Y^c})[A]}\leq {\varepsilon}.
  \end{align}
\end{lemma}
\begin{proof}
  The reduction map $\Gamma_{Y^c}^\Lambda$ onto the complement of $Y$ can be implemented as a twirl over the unitary group on $Y$, i.e.
  \begin{align}
	\Gamma_{Y^c}^\Lambda[A] = \int \rmd U_{Y}\, ( \id_{Y^c} \otimes U)A( \id_{Y^c} \otimes U^*),
\end{align}
where the integral runs over the normalized Haar measure on $Y$. Accordingly, the norm difference of $A$ and its reduction to $Y^c$ can be upper bounded by
\begin{align}
  \norm{(1-\Gamma_{Y^c})[A]} = \norm{ A -  \int \rmd U_{Y}\, ( \id_{Y^c} \otimes U)A( \id_{Y^c}   \otimes U^*)}
	\leq \int \rmd U_{Y}\, \norm{[A,\left(\id_{Y^c}\otimes U\right)] \left(\id_{Y^c}\otimes U^*\right)}\leq \varepsilon,
\end{align}
where in the last step, we used the unitary invariance of the operator norm, our assumption on $A$ and the normalization of the Haar measure.
\end{proof}

\section{From LR bounds of single-site observables to general LR bounds}
\label{app:LRbasis}
We have emphasized in the main text that our result implies that LR-bounds for single-site observables imply LR-bounds for observables with arbitrary support. 
Indeed, this follows immediately by combining our main result with theorem~\ref{thm:LRbounds}. 
One may be tempted to think that this follows directly from~\ref{thm:LRbounds}, since any operator $A$ supported on region $X$ may be expressed as
\begin{align}
    A = \sum_{i_1,\ldots,i_{|X|}} c_{i_1 \cdots i_{|X|}} S_{i_1}\otimes\cdots\otimes S_{i_{|X|}},
\end{align}
where the $S_{i}$ with $i=1,\ldots,d^2$ provide an operator basis for the space of operators acting on a single site of the lattice. Now using the Leibniz-rule for the commutator and the triangle inequality, a bound for single-site observables implies a bound for more general observables. However, it is important to realize that this procedure yields an additional pre-factor $d^{2|X|}$ on the r.h.s. of the LR-bound. For large regions $X$, for example, half the system, this pre-factor diverges exponentially quickly for large systems, making the resulting bound useless. 
Our result, on the other hand, shows that a Lieb-Robinson bound for single-site observables implies a corresponding bound for arbitrary observables without such a pre-factor.

 \section{Fermionic systems}
In the main text, we discussed the setting of a lattice of spins. Here we explain how our results transfer to fermionic lattice systems.
While the basic idea and proof-strategy are exactly identical, we need to be a bit more careful due to the parity super-selection rule for fermions.
To do this, let us first briefly recapitulate the formalism of fermionic lattice systems (see, e.g., Ref.~\cite{BratteliRobinson2} for an introduction to the mathematical formalism of fermionic systems). 
For recent Lieb-Robinson bounds in the context of fermionic lattice systems, see Ref.~\cite{Nachtergaele2017}, which also includes a detailed discussion of the formalism of fermionic lattice systems from which we took inspiration. However, our discussion of the reduction map is quite different and might be of independent interest.

\subsection{The basic set-up of fermionic lattice systems}
To every point of our lattice $\ZZ^D$, we associate an index-set $I_{x}=\{1,\ldots, d\}$.
For every index $j\in I_{x}$ we define annihilation and creation operators $f_{x,j}^{\phantom{\dagger}}, f_{x,j}^\dagger$ fulfilling the canonical anti-commutation relations:
\begin{align}
	\big\{f_{x,j}^{\phantom \dagger},f_{y,k}^\dagger\big\} = \delta_{x,y}\delta_{j,k}\id,\quad \big\{f_{x,j},f_{y,k}\big\} = 	\big\{f_{x,j}^\dagger,f_{y,k}^\dagger\big\} = 0.
\end{align}
The algebra $\mc A_X$ of local operators associated to a finite region $X$ is then generated by arbitrary monomials of the $f_{x,j}^{\phantom\dagger}, f_{y,k}^\dagger$ with $x,y\in X$. Every $\mc A_X$ is isomorphic to the matrix-algebra $M_{2^{d |X|}}$ of $2^{d|X|}\times 2^{d|X|}$ matrices.
A special role is played by the operator
\begin{align}
	P_{X} = \prod_{x\in X} \prod_{j\in I_x} (-1)^{n_{x,j}},
\end{align}
where $n_{x,j} := f_{x,j}^\dagger f_{x,j}^{\phantom \dagger}$ is the \emph{number operator} associated to the \emph{mode} $(x,j)$. 
Due to the canonical anti-commutation relations the ordering of the operators in the definition of $P_X$ does not matter. 
$P_X$ is called the \emph{parity operator} associated to the region $X$ and fulfills $P_X^2=\id$ and $P_X^\dagger  = P_X$.
Every $A\in \mc A_X$ can be decomposed into an \emph{even} and \emph{odd} part as
\begin{align}
	A^+ = \frac{A+P_X A P_X}{2},\quad A^- = \frac{A-P_X A P_X}{2}.
\end{align}
Then $[A^+,P_X]=0$ and $\{A^-,P_X\}=0$. We write
\begin{align}
	\mc A_X = \mc A_X^+ + \mc A_X^-,\quad	\mc A^+_X := \left\{ A^+ \ |\ A \in \mc A_X\right\},\quad \mc A^-_X := \left\{A^- \ |\ A\in \mc A_X \right\}.
\end{align}
Due to the parity super-selection rule of fermionic systems, all \emph{physical observables} are even self-adjoint operators, and hence part of some $\mc A_X^+$. We note that $\mc A_X^+$ is an algebra containing the identity, but $\mc A_X^-$ is not an algebra, since $a,b\in \mcA_X^-$ implies $ab \in \mc A_X^+$.
Alternatively, we can characterize $\mc A^+_X$ as being generated by monomials of an even number of the $f_{x,j}^{\phantom\dagger}, f_{y,k}^\dagger$ with $x,y\in X$, while $\mc A^-_X$ is generated by monomials of an odd number of the creation and annihilation operators. 
Due to the parity super-selection rule, the algebra of physical operators $\mc A^+_X$ is \emph{not} isomorphic to a full matrix-algebra, but isomorphic to a direct sum
\begin{align}
	\mc A^+_X \simeq M_{2^{d|X|-1}}\oplus M_{2^{d|X|-1}} \simeq M_{2^{d|X|-1}}\otimes \proj{0} + M_{2^{d|X|-1}}\otimes \proj{1}.
\end{align}
where the two direct summands correspond to the $+1$ and $-1$ eigenspaces of $P_X$, whose projectors we denote by $P_X^\pm$ and, which, in the above decomposition, are given by $P_X^+ = \id\otimes \proj{0}$, $P_X^- = \id\otimes \proj{1}$. Alternatively, we can write them as $P_X^\pm = (\id \pm P_X)/2$. Any operator $A\in \mc A_X$ can be decomposed as
\begin{align}
	A = (P_X^+ + P_X^-) A (P_X^+ + P_X^-) = \underbrace{A_{++} + A_{--}}_{A^+} + \underbrace{A_{+-} + A_{-+}}_{A^-}
\end{align}
with $A_{++} = P_X^+ A P_X^+,\ A_{+-} = P_X^+ A P_X^-$ and so forth.
For any algebra $\mc A_X$ there is a unique \emph{even tracial state} $\omega^{\Tr{}}$, i.e., a positive and normalized linear function, which vanishes on $\mc A_X^-$ and fulfills
\begin{align}
	\omega^{\Tr{}}[AB] = \omega^{\Tr{}}[A]\omega^{\Tr{}}[B],
\end{align}
for $A\in \mc A_{X_1}, B \in \mc A_{X_2}$ with $X_1\cap X_2=\emptyset$, and for any two $A,B\in \mc A_X$ we have $\omega^{\Tr{}}[AB]= \omega^{\Tr{}}[BA]$.

\subsection{The reduction map}
We also need to define a reduction map as in the case of spin-systems. Here, an additional difficulty arises in the setting of fermionic systems. To understand this issue, we should remember that, in any given physical situation, we should be allowed to add auxiliary systems to our description which do not take part in the time-evolution of the system of interest. For spin systems, such "innocent bystanders" do not make any difference. However, for fermionic systems, the situation is different due to the parity super-selection rule and the canonical anti-commutation relations. 
Anticipating this issue, let us therefore consider, next to the lattice of fermions, also another system $S$ of $|S|$ fermions, described by associated annihilation and creation operators $f^\dagger_s,f^{\phantom\dagger}_s$ for $s=1,\ldots, S$ with $\{f^\dagger_{x,j}, f^{\phantom\dagger}_s\}=0$ and $\{f_s, f_{x,j}\}=0$ for any $x\in \ZZ^D$. For $X\subseteq \Lambda\subset \ZZ^D$, we then introduce a reduction map
\begin{align}
	\Gamma^{\Lambda S}_X: \mc A_{\Lambda S} \rightarrow \mc A_{\Lambda S},
\end{align}
that maps physical operators on $\Lambda$ to physical operators on $X$, $\Gamma^{\Lambda S}_X:\mc A_{\Lambda}^+ \rightarrow \mc A_{X}^+$, and fulfills for any $A\in \mc A_{\Lambda S}^+$ (with $X^c=\Lambda\setminus X$)
\begin{align}
	\left[\Gamma^{\Lambda S}_X[A], B\right]=0\quad\forall B\in \mc A_{X^c S}^+.
\end{align}
Importantly, we will see later that the map $\Gamma^{\Lambda S}_X:\mc A_{\Lambda}^+ \rightarrow \mc A_{X}^+$ does not depend on $S$ as long as $|S|>1$. 

Let us now show how the reduction map is constructed. As shown in the proof of Lemma~\ref{lemma:reductionspins}, in the spin-case the reduction map can be written as
\begin{align}
	\Gamma^\Lambda_X[A] = \int \mathrm{d}U_{X^c}\ (\id_{X}\otimes U) A (\id_{X}\otimes U^\dagger), \quad\text{(Spins)}
\end{align}
where $\mathrm{d}U_{X^c}$ denotes the normalized Haar-measure on $X^c$. For the fermionic case, the construction we need is more involved.
First note that we can decompose
\begin{align}
	\mc A^+_{X^c S}	&\simeq (\mc A^+_{X^c}\otimes \mc A^+_S)\oplus(\mc A^-_{X^c}\otimes \mc A^-_S) \simeq M_{2^{d|X^c|+S-1}}\otimes \proj{0} + M_{2^{d|X^c|+S-1}}\otimes \proj{1},\\
\end{align}
where the last decomposition is with respect to the parity operator $P_{X^cS}$.
The group $\mc U_{X^c S}^+$ of \emph{even unitaries} on $X^c S$ thus corresponds to $U\big(2^{d|X^c|+S-1}\big)\oplus U\big(2^{d|X^c|+S-1}\big)$.
We now define the reduction maps as
\begin{align}
\Gamma^{\Lambda S}_X[A] &= \int \mathrm{d}U\, U A U^\dagger
\end{align}
where $\mathrm{d}U$ is the normalized Haar measure over $\mc U_{X^cS}^+$.
Since this group corresponds to two copies of the group $U\big(2^{d|X^c|+S-1}\big)$, one for each parity sector on $X^cS$, we can re-write this map as
\begin{align}
	\Gamma^{\Lambda S}_X[A] &= \int \mathrm{d}U_+ \int \mathrm{d} U_-\, (U_+ P_{X^cS}^++ U_- P_{X^cS}^-)A(U_+ P_{X^cS}^++ U_- P_{X^cS}^-)^\dagger \\
	&= \int \mathrm{d}U_+\, U_+ P_{X^cS}^+ A P_{X^cS}^+U_+^\dagger +  \int \mathrm{d}U_-\, U_- P_{X^cS}^- A P_{X^cS}^-U_-^\dagger \\
	&+ \left[\int \mathrm{d}U_+\, U_+\,  P_{X^cS}^+ AP_{X^cS}^- \int \mathrm{d}U_-\, U_-^\dagger  + \text{h.c.}\right]\\
	&= \int \mathrm{d}U_+\, U_+ P_{X^cS}^+ A P_{X^cS}^+U_+^\dagger +  \int \mathrm{d}U_-\, U_- P_{X^cS}^- A P_{X^cS}^-U_-^\dagger,
\end{align}
where $\mathrm{d}U_{\pm}$ are normalized Haar-measures over $U\big(2^{d|X^c|+S-1}\big)$ and we used that $\int \mathrm{d}U_{\pm} U_{\pm} =0$. For the Haar-measure, we have
\begin{align}
	\int \mathrm{d}U_+\, U_+ P_{X^cS}^+ A P_{X^cS}^+U_+^\dagger = \frac{\Tr_{P_{X^cS}^+}\big[ A P_{X^cS}^+\big]}{2^{d|X^c|+|S|-1}}P_{X^c S}^+ = 2 \omega^{\Tr{}}_{X^c S}\big[A P_{X^cS}^+\big]P_{X^c S}^+,
\end{align}
where $\omega^{\Tr{}}_{X^cS}$ denotes the \emph{partial} tracial state or maximally mixed state on $X^c S$. On operators of the form $A_XB_{X^c S}$, it acts as
\begin{align}
	\omega^{\Tr{}}_{X^cS}[A_XB_{X^cS}] = A_X \omega^{\Tr{}}[B_{X^cS}].
\end{align}
Note that wenn an operator is supported in the support of the tracial state, we can omit the subscript. We summarize this result in the following Lemma.

\begin{lemma}\label{lemma:restr_props} The reduction map is a unital, completely positive map, which can be written as
	\begin{align}
		\Gamma^{\Lambda S}_X[A] &= 2\left(\omega^{\Tr{}}_{X^cS}\big[A P_{X^c S}^+\big]P_{X^cS}^++\omega^{\Tr{}}_{X^cS}\big[A P_{X^c S}^-\big]P_{X^cS}^-\ \right).
	\end{align}
	It fulfills:
	\begin{enumerate}
		\item If $A,C\in \mc A^+_X$ and $B\in \mc A^+_{X^c S}$, we have $\Gamma^{\Lambda S}_X[ABC] = A\Gamma^{\Lambda S}_{X}[B]C$.
		\item If $B\in \mc A_{X^c S}^-$, then $\Gamma^{\Lambda S}_{X}[B]=0$.
		\item If $|S|\geq 1$, the restriction to $\mc A_\Lambda$ maps to $\mc A_X^+$, $\Gamma^{\Lambda S}_X:\mc A_\Lambda\rightarrow \mc A_X^+$, and the result does not depend on $S$.
		\item For $X\subseteq \Lambda , Y\subseteq \Lambda$ and when restricted to $\mc A_\Lambda^+$,  we have
			\begin{align}
				\Gamma^{\Lambda S}_X \circ \Gamma^{\Lambda S}_Y =\Gamma^{\Lambda S}_Y \circ \Gamma^{\Lambda S}_X = \Gamma^{\Lambda S}_{X\cap Y}.
			\end{align}
	\end{enumerate}
	\begin{proof}
	1. is obvious. 2. follows from the fact that $\omega^{\Tr{}}_{X^c S}$ is an even state and hence vanishes on odd operators.
	For 3. we make the following observations. First, if $A\in\mc A_\Lambda$ is in $\mc A_{X}\otimes \mc A_{X^c}^-$, then the reduction vanishes by 1. and 2.
	Therefore consider operators of the form $A_X A_{X^c}$ with $A_{X^c}\in \mc A_{X^c}^+$. By 1. we can forget about the $A_X$ part and in the following consider $A\in \mc A_{X^c}^+$.
		We use that for $|S|\geq 1$, we have
		\begin{align}
			P_{X^c S}^+ &= \frac{\id + P_{X^c S}}{2} = \frac{1}{2}\left( \big(P_{X^c}^++ P_{X^c}^-\big) \big(P_{S}^++ P_{S}^-\big) + \big(P_{X^c}^+- P_{X^c}^-\big) \big(P_{S}^+- P_{S}^-\big)\right) = P_{X^c}^+ P_S^+ + P_{X^c}^- P_S^-,\\
			P_{X^c S}^- &= \frac{\id - P_{X^c S}}{2} = \frac{1}{2}\left( \big(P_{X^c}^++ P_{X^c}^-\big) \big(P_{S}^++ P_{S}^-\big) - \big(P_{X^c}^+- P_{X^c}^-\big) \big(P_{S}^+- P_{S}^-\big)\right) = P_{X^c}^+ P_S^- + P_{X^c}^- P_S^+.
		\end{align}
		Since $P_{X^c}^+ P_{X^c S}^\pm = P_{X^c}^+ P_S^\pm$ and $P_{X^c}^- P_{X^cS}^\pm = P_{X^c}^- P_S^\mp$, we have
		\begin{align}
			AP_{X^c S}^\pm = A (P_{X^c}^+ + P_{X^c}^-)P_{X^c S}^\pm = A(P_{X^c}^+P_{S}^\pm + P_{X^c}^-P_S^\mp).
		\end{align}
		Using that $\omega^{\Tr{}}_{X^c S}$ is a product-state and $\omega^{\tr{}}_S\big[P_S^\pm\big]=1/2$, we find
		\begin{align}
			\omega^{\Tr{}}_{X^c S}\big[AP_{X^c S}^+\big] &= \omega^{\Tr{}}_{X^c}[AP_{X^c}^+]\omega^{\tr{}}_S\big[P_S^+\big] +  \omega^{\Tr{}}_{X^c}[AP_{X^c}^-]\omega^{\tr{}}_S\big[P_S^-\big]\\
			&=\frac{1}{2}\omega^{\Tr{}}_{X^c}\big[A(P_{X^c}^++P_{X^c}^-)] = \frac{1}{2}\omega^{\Tr{}}_{X^c}[A]
		\end{align}
		and
		\begin{align}
			\omega^{\Tr{}}_{X^c S}\big[AP_{X^c S}^-\big] &= \omega^{\Tr{}}_{X^c}[AP_{X^c}^+]\omega^{\tr{}}_S\big[P_S^-\big] +  \omega^{\Tr{}}_{X^c}[AP_{X^c}^-]\omega^{\tr{}}_S\big[P_S^+\big]\\
			&=\frac{1}{2}\omega^{\Tr{}}_{X^c}\big[A(P_{X^c}^++P_{X^c}^-)] = \frac{1}{2}\omega^{\Tr{}}_{X^c}[A].
		\end{align}
		Consequently,
		\begin{align}
			\Gamma^{\Lambda S}_X[A] = \omega^{\Tr{}}_{X^c}[A]\left(P_{X^cS}^+ + P_{X^c S}^-\right) = \omega^{\Tr{}}_{X^c}[A]\id\ \in \mc A_{X}^+.
		\end{align}
		Finally, we need to show 4.  We consider operators $A\in\mc A_{(X\cup Y)^c}, B\in \mc A_{X\setminus Y},C\in\mc A_{Y\setminus X}, D\in\mc A_{X\cap Y}$ such that $ABCD\in \mc A_\Lambda^+$.
		Then
		\begin{align}
			\Gamma^{\Lambda S}_{X\cap Y}[ABCD]=\Gamma^{\Lambda S}_{X\cap Y}[ABC]D = \omega^{\Tr{}}[ABC]D.
		\end{align}
		Similarly, we have
		\begin{align}
			\Gamma^{\Lambda S}_X\circ \Gamma^{\Lambda S}_Y[ABCD] =  \Gamma^{\Lambda S}_X\big[\Gamma^{\Lambda S}_Y[AB]CD] = \Gamma^{\Lambda S}_X[C]D\omega^{\Tr{}}[AB] = D\omega^{\Tr{}}[C] \omega^{\Tr{}}[AB] = \omega^{Tr}[ABC]D,
		\end{align}
		where we used that $\omega^{\Tr{}}$ is a product-state. For the reversed order:
		\begin{align}
			\Gamma^{\Lambda S}_Y\circ \Gamma^{\Lambda S}_X[ABCD]=\pm\Gamma^{\Lambda S}_Y\big[B\Gamma^{\Lambda S}_X[AC]D\big] = \pm \Gamma^{\Lambda S}_Y\big[B\big]D\omega^{\Tr{}}[AC] =\pm D \omega^{\Tr{}}[ACB],
		\end{align}
		where the negative sign only appears if both $A$ and $B$ are odd, in which case the expression vanishes. We thus conclude that all three expressions vanish if any of the $A,B,C$ are odd and otherwise give the same result. Hence they always give the same result.
	\end{proof}
\end{lemma}%
We emphasize that $\Gamma^{\Lambda S}_X$ does not map all observables to observables supported in $X$. For example, the parity operator $P_{X^c S}$ is mapped to itself. This is easy to see from the definition via the average over even unitaries, but can also be seen from the above form of the map using $P_{X^c S} = P_{X^c S}^+ - P_{X^c S}^-$:
\begin{align}
	\Gamma^{\Lambda S}_X[ P_{X^c S}^+ - P_{X^c S}^-] &= 2\left(\omega^{\Tr{}}_{X^cS}\big[P_{X^c S}^+\big]P_{X^cS}^+-\omega^{\Tr{}}_{X^cS}\big[P_{X^c S}^-\big]P_{X^cS}^-\ \right)\\
	&=2\left(\frac{1}{2}P_{X^cS}^+-\frac{1}{2}P_{X^cS}^-\ \right) = P_{S^cX}.
\end{align}
We close the discussion of the reduction map for fermionic lattice systems by proving the analogue of Lemma~\ref{lemma:reductionspins}:
\begin{lemma}
	Suppose $|S|\geq 1$ and let $A\in \mc A_\Lambda^+$ be such that
	\begin{align}
		\norm{[A,B]}\leq \eps\norm{B} \quad \forall B\in \mc A_{X^c S}^+.
	\end{align}
	Then $\Gamma^{\Lambda S}_X[A]\in \mc A_X^+$ and
	\begin{align}
	\norm{A-\Gamma^{\Lambda S}_X[A]}\leq \eps.
	\end{align}
	\begin{proof}
		That $\Gamma^{\Lambda S}_X[A]\in \mc A_X^+$ for $|S|\geq 1$ was shown in the previous Lemma. The second claim follows as in the proof of Lemma~\ref{lemma:reductionspins}:
		\begin{align}
			\norm{A - \Gamma^{\Lambda S}_X[A]} \leq \int_{\mc U_{X^c}^+} \mathrm{d}U\,\norm{A UU^\dagger - UAU^\dagger} =  \int_{\mc U_{X^c}^+} \mathrm{d}U\,\norm{[A,U]}\leq \eps,
		\end{align}
		where we used the triangle inequality and normalization of the Haar measure in the first step and unitary invariance of the norm in the second step.
	\end{proof}
\end{lemma}
As shown in Lemma~\ref{lemma:restr_props}, the map $\Gamma^{\Lambda S}_X:\mc A_\Lambda\rightarrow \mc A_X^+$ does not in fact depend on $S$. We can hence define
\begin{align}\label{eq:fermionic_restriction_map}
	\Gamma^{\Lambda}_X:\mc A_\Lambda\rightarrow \mc A_X^+,\quad A\mapsto \Gamma^{\Lambda S}_X[A],\quad \text{for some $|S|\geq 1$},
\end{align}
and this map fulfills all the relevant properties that we needed in the case of spins.

\subsection{Dynamics and the fermionic result}
Similarly as for spin-systems, we describe dynamics by a potential, which now associates \emph{even} operators $\Phi(X)=\Phi(X)^\dagger \in \mc A_X^+$ to finite regions $X$ of the lattice.
Therefore the Hamiltonian $H_\Lambda$ is even and the induced time-evolution $\tau_t^\Lambda$ maps $\mc A^\pm_\Lambda$ to $\mc A^\pm_\Lambda$ as required. Analogously to the spin-setting,  we call such a potential $k$-body if $\Phi(X)=0$ for $|X|>k$. As for spin-systems, exponentially decaying $k$-body potentials lead to Lieb-Robinson bounds also in the case of Fermions (see, e.g., Ref.~\cite{Nachtergaele2017} for a recent account). Similarly as for the reduction map, we explicitly take into account the possibility of fermionic auxiliary systems $S_1,S_2$ that are left invariant under the time-evolution. This allows us to recover our result also in the case of fermionic lattice systems:
\begin{theorem}[Fermionic converse to Lieb-Robinson bounds]\label{thm:converse_fermions}
	Consider a fermionic lattice system with $k$-body potential and assume that there exist a function $h(t,x)$ with $h(0,x)=0$, with $h(0,x)=0$ for $|x|>0$ and differentiable in $t$, such that for any finite subset $\Lambda\subset \ZZ^D$, any two points $x\neq y\in \ZZ^D$ and any two fermionic systems $S_1,S_2$ we have
	\begin{align}
		C^\Lambda_{x,y}(t)&:=\max_{A,B} \frac{\norm{[\tau^\Lambda_t(A),B]}}{\norm{A}\norm{B}}  \leq h(t,|x-y|),
	\end{align}
	where the maximization is over physical operators on $\{x\}\cup S_1$ and $\{y\}\cup S_2$: $A\in \mc A_{\{x\}\cup S_1}^+, B\in \mc A_{\{y\}\cup S_2}^+$.
	Then there exists a constant $m_k$ only depending on $k$, and a $k$-body potential $\hat \Phi_\Lambda$ on $\Lambda$, giving rise to the same time-evolution $\tau^\Lambda_t$, with
	\begin{align}
		\norm{\hat\Phi_\Lambda(X)} \leq m_k\, \partial_{\!t}\! h\big(0,\mathrm{diam}(X)\big).
	\end{align}
	for all $X\subseteq\Lambda$ with $|X|\geq 2$.
\begin{proof}
	The proof is essentially identical to the spin-case using the fermionic reduction map defined in \eqref{eq:fermionic_restriction_map}.
\end{proof}
\end{theorem}

\section{Strictly local, discrete-time, unitary dynamics that cannot be generated locally in continuous time}
Here, we briefly describe an example of discrete-time, unitary dynamics that is strictly local, but cannot be generated in continuous time with a (quasi-)local Hamiltonian.
This example is formulated in first-quantized language, but may be second-quantized to yield a fermionic lattice model.
The model describes a particle hopping on a one-dimensional line, which we can view as $\mathbb Z$, with Hilbert-space is given by $l^2(\mathbb Z)$. 
An orthonormal basis is given by associating one basis-vector $\ket{x}$ with each lattice site $x\in \mathbb Z$, which is interpreted as the particle being located at the site $x$.
The unitary dynamics that we consider simply amounts to shifting the particle by one site:
\begin{align}
	U\ket{x} = \ket{x+1}.
\end{align}
In other words, $U = \exp(\mathrm i P)$, where $P$ is the lattice momentum, which generates translations.
It's clear that this discrete-time dynamics is strictly local.
However, the matrix-elements of the lattice-momentum in position space fulfill (for $x\neq y$)
\begin{align}
	\big|\! \langle y| P|x\rangle \big| = \frac{1}{|y-x|}.
\end{align}
Therefore, if we choose $A=\proj{x},B=\proj{y}$ for $x\neq y$, we find that for small times $t$:
\begin{align}
	\norm{\Big[\e^{\mathrm i P t}A \e^{-\mathrm i P t},B\Big]} = \frac{t}{|x-y|} + \mathcal O(t^2).
\end{align}
Thus, while the dynamics is strictly local for discrete times, it becomes long-ranged for intermediate times.

\end{document}